\documentclass[a4paper,twoside]{article}
\usepackage[latin1]{inputenc} %
\usepackage[T1]{fontenc} %
\usepackage{rr-orange}
\usepackage{amsmath,amssymb,amsthm}
\usepackage{subfig}
\usepackage{hyperref}
\usepackage{algorithmic}
\usepackage{algorithm}
\usepackage{multirow}

\newtheorem{theorem}{Theorem}
\newtheorem{conjecture}{Conjecture}

\date{September 2009}

\author{
Fabien Mathieu (Orange Labs, France)%
}
\rauthor{Fabien Mathieu}
\title{Heterogeneity in\\ Distributed Live Streaming:\\ Blessing or Curse?}
\rtitle{Heterogeneity in Distributed Live Streaming}

\RRNo{RR-OL-2009-09-001}
\version{1}

\begin{document}
\maketitle   

\begin{executive}
Distributed live streaming has brought a lot of interest in the past few years. In the homogeneous case (all nodes having the same capacity), many algorithms have been proposed, which have been proven almost optimal or optimal. On the other hand, the performance of heterogeneous systems is not completely understood yet.

In this paper, we investigate the impact of heterogeneity on the achievable delay of chunk-based live streaming systems. We propose several models for taking the atomicity of a chunk into account. For all these models, when considering the transmission of a single chunk, heterogeneity is indeed a ``blessing'', in the sense that the achievable delay is always faster than an equivalent homogeneous system. But for a stream of chunks, we show that it can be a ``curse'': there is systems where the achievable delay can be arbitrary greater compared to equivalent homogeneous systems. However, if the system is slightly bandwidth-overprovisionned, optimal single chunk diffusion schemes can be adapted to a stream of chunks, leading to near-optimal, faster than homogeneous systems, heterogeneous live streaming systems.
\end{executive}

\tableofcontents
\newpage

\section{Introduction}

Recent years have seen the proliferation of live streaming content diffusion over the Internet. In order to manage large audience, many distributed scalable protocols have been proposed and deployed on peer-to-peer or peer-assisted platforms~\cite{coolstreaming,pplive,sopcast,tvants,uusee}. Most of these systems rely on a chunk-based architecture: the stream is divided into small parts, so-called chunks, that have to be distributed independently in the system.

The measurements performed on distributed P2P platforms have shown that these platforms are highly heterogeneous with respect to the shared resources, especially the upload bandwidth~\cite{saroiu2002measurement,conf/globecom/CiulloMML08}. However, except for a few studies (see for instance~\cite{liu07minimum,liu08performance}), most of the theoretical research has been devoted to the analysis of homogeneous systems, where all the peers have similar resources.

At first sight, it is not clear whether heterogeneity should be positive or negative for a live streaming system. On the one hand, some studies on live streaming algorithms have reported a degradation of the performance when considering heterogeneous scenarios~\cite{bonald08epidemic}. On the other hand, consider these two toy scenarios:
\paragraph*{Homogeneous} a source injects a live stream at a rate of one chunk per second into a system of $n$ peers, each peer having an upload bandwidth of one chunk per second. Then the best achievable delay to distribute the stream is $\lceil \log_2(n)\rceil$ seconds~\cite{bonald08epidemic}.
\paragraph*{Centralized} same as above, except that one peer has an upload bandwidth of $n$ chunks per second, and the others have no upload capacity. Then the stream can obviously be distributed within one second.

The total available bandwidth is the same in both scenarios, and the centralized one can be seen as an extremely heterogeneous ``distributed'' scenario, so this simple example suggests that heterogeneity should improve the performance of a live streaming system.

In this paper, we propose to give a theoretical background for the feasible performance of distributed, chunk-based, heterogeneous, live streaming systems. The results proposed here are not meant to be directly used in real systems, but they are tight explicit bounds, that should serve as landmarks for evaluating the performance of such systems, and that can help to understand if heterogeneity is indeed a ``blessing'' or a ``curse'', compared to homogeneity.

\subsection{Contribution}

We propose a simple framework for evaluating the performance of chunk-based live streaming systems. Several variant are proposed, depending on whether multi-sources techniques are allowed or not, and on the possible use of parallel transmissions. For the problem of the optimal transmission of a single chunk, we give the exact lower bounds for all the considered variants of the model. These bounds are obtained either with an explicit closed formula or by means of simple algorithms. Moreover, the bounds are compared between themselves and to the homogeneous case, showing that heterogeneity is an improvement for the single chunk problem. For the transmission of a stream of chunks, we begin by a feasibility result that states that if there is enough available bandwidth, a system can achieve lossless transmission within a finite delay. However, we provide very contrasted results for the precise delay performance of such systems: on the one hand, we show that there are bandwidth-over-provisioned systems that need a $\Omega(N)$ transmission delay, whereas equivalent homogeneous systems only need $O(\log(N))$; on the other hand, we give simple, sufficient conditions that allows to relate the feasible stream delay to the optimal single-chunk delay.

The rest of the paper is organized as follows. Section~\ref{sec:model} presents our model and notation, and Section~\ref{sec:related} presents the related work. Then Section~\ref{sec:single} presents the bounds for the diffusion of one single chunk, while Section~\ref{sec:stream} considers the case of a stream of chunks. Section~\ref{sec:conclu} concludes.

\section{Model}
\label{sec:model}

We consider a distributed live streaming system consisting of $N$ entities, called peers. A source injects some live content into the system, and goal is that all peers receive that content with a minimal delay. We assume no limitation on the overlay graph, so any peer can potentially transmit a chunk to any other peer (full mesh connectivity).

\subsection{Chunk-based diffusion}
\label{model:chunk}

The content is split into atomic units of data called \emph{chunks}. Chunk decomposition is often used in distributed live streaming systems, because it allows more flexible diffusion schemes: peers can exchange maps of the chunks they have/need, and decide on-the-fly of the best way to achieve the distribution. The drawback is the induced data quantification. Following a standard approach~\cite{bonald08epidemic}, we propose to model this quantification by assuming that a peer can only transmit a chunk if it has received a complete copy of it.

For simplicity, we assume that all chunks have the same size, which we use as data unit.

\subsection{Capacity constraints}
\label{model:capacity}

We assume an upload-constrained context, where the transmission time depends only upon the upload bandwidth of the sending peer: if a peer $i$ has upload bandwidth $u_i$ (expressed in chunks per second), the transmission time for $i$ to deliver a chunk to any other peer is $\frac{1}{u_i}$. Without loss of generality, we assume that the peers are sorted decreasingly by their upload bandwidths, so we have $u_1\geq u_2\geq ... \geq u_N\geq 0$. We also assume that the system has a non-null upload capacity ($u_1>0$).

For simplicity, we assume that there is no constraint on the download capacity of a peer, but we will discuss the validity of that assumption later on.

\subsection{Collaborations}
\label{model:collaborations}

We also need to define the degree of collaboration enabled for the diffusion of one chunk, i.e. how many peers can collaborate to transmit a chunk to how many peers simultaneously. The main models considered in this paper are:
\paragraph{Many-to-one (short notation: $(\infty/1)$)} The $(\infty/1)$ model allows an arbitrary number of peers to collaborate when transmitting a chunk to a given peer message. three peers $i$, $j$, $k$ can collaborate to transmit a chunk they have to a fourth peer in a time $\frac{1}{u_i+u_j+u_k}$. The $(\infty/1)$ model may not be very practical, because it allows $N-1$ peers to simultaneously collaborate for one chunk, which can generate synchronization issues and challenge the assumption that download is not a constraint (the receiving peer must handle the cumulative bandwidths of the emitters). However, it has a strong theoretical interest, as it encompasses more realistic models. Therefore the $(\infty/1)$ bounds can serve as landmark for the other models.
\paragraph{One-to-one (short notation: $(1/1)$)}
In the $(1/1)$ model, a chunk transmission is always performed by a single peer: if at some time, three peers $i$, $j$ and $k$ have a chunk and want to transmit it, they must select three distinct receivers, which will receive the message after $\frac{1}{u_i}$, $\frac{1}{u_j}$ and $\frac{1}{u_k}$ seconds respectively. The connectivity and download bandwidth burdens are considerably reduced in that model. Note that $(1/1)$ is included in $(\infty/1)$ (any algorithm that works under $(1/1)$ is valid in $(\infty/1)$).
\paragraph{One-to-some (short notation: $(1/c)$)} The models above implicitly assume that a given peer transmits chunks sequentially, but for technical reasons, practical systems often try to introduce some parallelism in the transmission process: pure serialization can lead to a non-optimal use of the sender's transmission buffer, for instance in case of connectivity or node failures. We propose the $(1/c)$ model to take parallelism into account: a transmitting peer $i$ always splits its upload bandwidth into $c$ distinct connections of equal capacity. We model a price for the use of parallelism, by assuming that these connections cannot be aggregated. That means that a peer $i$ can transmit to up to $c$ receivers simultaneously, but it always needs $\frac{c}{u_i}$ seconds to transmit the message to any given peer. Note that any algorithm that works in the $(1/c)$ model can be emulated in the $(1/1)$ model.

\subsection{Single chunk / stream of chunks diffusion delays}
\label{model:singlestream}

In order to study the achievable diffusion delay of the system, we propose a two step approach: we first consider the feasible delay for the transmission of a single chunk, then we investigate how this can be related to the transmission delay of a stream of chunks.

In the single chunk transmission problem, we assume that at time $t=0$, $n_0$ copies of a newly created chunk are delivered to $n_0$ carefully selected distinct peers ($1\leq n_0 \leq N$), and we want to know the minimal delay $D(n)$ needed for $n$ copies of the chunk to be available in the system.  Note that as the system has a non-null upload capacity, $n$ copies can always be made in a finite time, so $D$ is well defined. The main value of interest is $D(N)$ (time needed for all peers to get a copy of the chunk), but $n>N$ can also be considered for theoretical purposes (we assume then that the extra copies are transmitted to dummy nodes with null upload capacity). We use the notation $D_m$, $D_1$ or $D_c$ depending on the model used (many-to-one, one-to-one or one-to-$c$ respectively).

In the stream of chunk problem, new chunks are created at a given rate $s$ (expressed in chunks per second) and injected with redundancy $n_0$. In other words, every $\frac{1}{s}$ seconds, $n_0$ copies of a newly created chunk are delivered to $n_0$ carefully selected peers. The (possibly infinite) stream is \emph{feasible} if there is a diffusion scheme that insures a lossless transmission within a bounded delay. It means that there is a delay such that \emph{any} chunk, after being injected in the system, is available to the $N$ peers within that delay. For a given feasible stream, we call $\tilde{D}$ (or $\tilde{D_m}$, $\tilde{D_1}$, $\tilde{D_c}$ if the underlying model must be specified) the corresponding minimal achievable delay. Obviously, $D$ is a lower bound for $\tilde{D}$.

\section{Related work}
\label{sec:related}

The problem of transmitted a message to all the participants (broadcast) or a subset of it (multicast) in a possibly heterogeneous capacity-constraint environment is not new. A few years ago, so-called \emph{networks of workstations} have been the subject of many studies~\cite{hall98scheduling,banikazemi98efficient,hadas01multicast,fraigniaud05efficient}. However, most of the results presented in those studies were too generic for presenting a direct interest for the chunk-based live streaming problem.

As far as we know, the work that is probably the closest to ours has been made by Yong Liu~\cite{liu07minimum}. For the single chunk problem, Liu has computed $D_1$ in specific scenarios, and he gave some (non tight) bounds for the general case. For the stream problem, he gave some insight on the delay distribution when the capacities are random, independent variables. Liu's study is more complete than ours for specific scenarios and implementation, but we provide tighter results for the general case, where no assumption is made but the $(*/*)$ model.

There is also two closely related problems for which theoretical analysis and fundamental limitations have been considered: the chunk-based, homogeneous, live streaming problem and the stripe-based, possibly heterogeneous live streaming problem.

For chunk-based homogeneous systems, the main result is that if the peers have enough bandwidth to handle the streamrate ($u\geq s$), then the stream problem is feasible for the $(1/1)$ model and we have $D_1=\tilde{D}_1=\frac{1}{u}\log_2(\frac{N}{n_0})$ (see for instance~\cite{liu07minimum}). The intuitive idea is that as all peers have the same bandwidth values, they can exchange their place in a diffusion tree without changing the performance of that tree. This allows to use the optimal diffusion tree for each new chunk introduced in the system: when a chunk is an internal node of the tree of a given chunk $i$, he just have to be a leaf in the trees of the next nodes until the diffusion of $i$ is complete. Of course, this \emph{permutation} technique cannot be used in a heterogeneous case.

The stripe-based model consists in assuming that the stream of data can be divided into arbitrary small sub-streams, called stripes. There is no chunk limitation in that model, therefore the transmission of data between nodes is only delayed by latencies. The upload capacity is still a constraint, but it only impacts the amount of stream that a peer can relay. A pretty complete study for the performance bounds of stripe-based systems is available in~\cite{liu08performance}. It shows that as long as there is enough bandwidth to sustain the stream (meaning, with our notation, $n_0+\frac{1}{s}\sum_{i=1}^Nu_i\geq N$), the stream can be diffused within a minimal delay. In Section~\ref{sec:stream}, we will show that this feasibility result can be adapted to the chunk-based model, although the delay tends to explode in the process.

\section{Single chunk diffusion}
\label{sec:single}

As expressed in \S~\ref{model:singlestream}, $D$ is a lower bound for $\tilde{D}$, so it is interesting to understand the single chunk problem. Moreover, as we will see in the next section, an upper bound for $\tilde{D}$ can also be derived from $D$ on certain conditions.

\subsection{\texorpdfstring{$(\infty/1)$ diffusion}{(infinity/1) diffusion}}

We first consider the many-to-one assumption, where collaboration between uploaders is allowed. Under this assumption, we can give an exact value for the minimal transmission delay.

\begin{theorem}
\label{thm:dmin}
Let $U_k$ be the cumulative bandwidth of the $k$ best peers ($U_k=\sum_{i=1}^ku_i$). Then the minimal transmission delay $D_m$ is given by
\begin{equation}
	D_{m}(n)=\sum_{k=n_0}^{n-1}\frac{1}{U_{k}}\text{.}
	\label{eq:dmin_multi}
\end{equation}
\end{theorem}

\begin{proof}
We say that a given peer is \emph{capable} when it owns a complete copy of the chunk (it is capable to tranmist that chunk). If at a given time the sum of upload bandwidths of the \emph{capable} peers (i.e. with a complete copy of the message) is $U$, then the minimal time for those peer to send a complete copy of the chunk to another peer is $\frac{1}{U}$. From that observation, we deduce that maximizing $U$ during the whole diffusion is the way to obtain minimal transmission. This is achieved by injecting the $n_0$ primary copies of the message to the $n_0$ best peers, then propagating the message peer by peer, always using all the available bandwidth of capable peers and selecting the target peers in decreasing order of upload. This gives the bound.
\end{proof}

\paragraph*{Remark} in~\cite{liu07minimum}, Liu proposed $D_m$ as a (loose) lower bound for $D_1$. Indeed, $D_m$ is an absolute lower bound for any chunk-based system, because the diffusion used makes the best possible use of the available bandwidth at any time. The only way to go below $D_m$ would be to allow peers to transmit partially received chunks, which is contrary to the chunk-based main assumption. Thus $D_m$ can serve as a reference landmark for all the delays considered here. Moreover, an appealing property of $D_{m}$ is that it is a direct expression of the bandwidths of the system, so it is straightforward to compute as long as the bandwidth distribution is known.

\subsubsection{Homogeneous case}

If all peers have the same upload bandwidth $u_i=u$, we have $U_k=ku$ for $k\leq N$, so the bound $D_m$ becomes simpler to express for $n\leq N$:
\begin{equation}
	D_m(n)= \frac{1}{u} \sum_{k=n_0}^{n-1}\frac{1}{k}\text{.}
	\label{eq:delay_dm}
\end{equation}

In particular, for $N\geq n\gg n_0$, the following approximation holds:

\begin{equation}
	D_m(n)\approx\frac{\ln (\frac{n}{n_0})}{u}\text{.}
	\label{eq:delay_dm2}
\end{equation}

So in the homogeneous case, the $(\infty/1)$ transmission delay is inverse proportional to the common upload bandwidth, and grows logarithmically with the number of peers.

\subsubsection{Gain of heterogeneity}

We can compare the performance of a given heterogeneous system to the homogeneous case: let us consider a heterogeneous system with average peer bandwidth $\bar{u}$, and maximum bandwidth $u_{\max}$. As peers are sorted by decreasing bandwidth, we have $k\bar{u}\leq U_k \leq ku_{\max}$. From \eqref{eq:dmin_multi}, it follows that
\begin{equation}
D_m^{u_{\max}} \leq D_{m}\leq D_m^{\bar{u}}\text{,}
\label{eq:dm_gain}
\end{equation}
\noindent where $D_m^u$ is $D_m$ in a homogeneous system with common bandwidth $u$. In particular, by combining the previous equations, one gets
\begin{equation}
D_m(n)<\frac{1}{\bar{u}}(\ln(\frac{n-1}{n_0})+\frac{1}{n_0})\text{.}
\label{eq:gain_dm}
\end{equation}
\noindent In other words, the optimal transmission delay is smaller for a heterogeneous system than for a homogeneous system with same average peer upload bandwidth. In that sense, heterogeneity can be seen as a ``blessing'' for the transmission of one single chunk.

\subsubsection{Homogeneous classes}

Equation~\eqref{eq:delay_dm2} can be extended to the case where there is classes of peers, each class being characterized by the common value of the upload bandwidths of its peers. 

\begin{theorem}
\label{thm:dmin_classes}
We assume here that we have $l$ classes with respective population size and upload bandwidth $(n_1,u_1)$,\ldots,$(n_k,u_l)$, with $u_1>\ldots >u_l$ and $n_i\gg 1$ (large population sizes). If $n_0\leq n_1$, then we have
\begin{equation}
	D_{m}(N)\approx \frac{1}{u_1}\ln(\frac{n_1}{n_0})+\sum_{i=2}^{l}\frac{\ln(1+\frac{n_iu_i}{\sum_{j=1}^{i-1}n_ju_j})}{u_i}\text{.}
	\label{eq:dmin_classes}
\end{equation}
\end{theorem}
\begin{proof}
because the minimal delay is obtained by transmitting the message to the best peers first, in the class scenario, the optimal transmission must follow the class order, beginning by the $(n_1,u_1)$ class and ending by the $(n_l,u_l)$ class. So in the minimal delay transmission, the $n_0$ initial messages are inserted in the first class and in a first phase, it will only be disseminated within that class. According to Equation~\eqref{eq:delay_dm2}, after about $\frac{1}{u_1}\ln(\frac{n_1}{n_0})$ seconds, all peers of the first class have a copy of the message.

Then, for the generic term of Equation~\eqref{eq:dmin_classes}, we just need to consider that the time $D_{i-1\rightarrow i}$ needed to fill up a class $i$, $2\leq i \leq l$, after all previous classes are already capable. $D_{i-1\rightarrow i}$ is given by Equation~\eqref{eq:dmin_multi}, with $n_0=\sum_{j=1}^{i-1}n_j$ (previous classes total size) and $n=\sum_{j=1}^{i}n_j$ (previous plus current classes size):

$$
\begin{array}{rl}
D_{i-1\rightarrow i} & =\sum_{k=\sum_{j=1}^{i-1}n_j}^{(\sum_{j=1}^{i}n_j)-1}\frac{1}{U_{k}}=\sum_{k=0}^{n_i-1}\frac{1}{U_{\sum_{j=1}^{i-1}n_j+k}}\\
& = \sum_{k=0}^{n_i-1}\frac{1}{(\sum_{j=1}^{i-1}n_ju_j)+ku_i}\\
& \approx \frac{1}{u_i}\ln(1+\frac{n_iu_i}{\sum_{j=1}^{i-1}n_ju_j})\text{.}
\end{array}$$

By summing $D(i)$ for $2\leq i\leq l$, one obtains the Equation~\eqref{eq:dmin_classes}.
\end{proof}

\paragraph*{Remark} if we have $n_1u_1\gg n_iu_i$ for all $2\leq i\leq l$ (case where the total upload capacity of the first class is far greater than the capacities of the other classes), we have a simpler approximation for $D_m$: 

\begin{equation}
D_{m}(N)\approx \frac{1}{u_1}\ln(\frac{n_1}{n_0})+\sum_{i=2}^{l}\frac{n_i}{\sum_{j=1}^{i-1}n_ju_j}\text{.}
\label{eq:dominant}
\end{equation}

In particular, if we consider, following~\cite{liu07minimum}, a two-class scenario, the second class being made of \emph{free-riders} ($u_2=0$), Equation~\eqref{eq:delay_dm2} simplifies into:
\begin{equation}
D_{m}(n)\approx \frac{1}{u}\ln(\frac{\min(n,n_1)}{n_0})+\frac{\max(n-n_1,0)}{Nu}\text{.}
\label{eq:freeriders}
\end{equation}
\noindent The first class gets the message after a logarithmic time, while it is linear for the free-rider class.

\subsection{\texorpdfstring{$(1/1)$, $(1/c)$ diffusion}{(1/1), (1/c) diffusion}}

In the diffusion scheme used for Theorem~\ref{thm:dmin}, all capable peers collaborate together to transmit the chunk to one single peer. Obviously, this approach is not sustainable because of the underlying cost for synchronizing an arbitrary great number of capable peers may be important anyhow and of the download bandwidth that the receiver peer must handle.

In practice, many systems do not rely on multi-sources capabilities and use one-to-one transmissions instead. We propose now to consider the minimal delay $D_1$ for the $(1/1)$, and compare it with the bound $D_m$.

Contrary to $D_m$, for which a simple closed formula exists, $D_1$ is hard to express directly. However, it is still feasible to compute its exact value, which is given by Algorithm~\ref{algo:r2sw}.
\begin{algorithm}
\caption{Algorithm to compute $D_1$}
\label{algo:r2sw}
\begin{algorithmic}[1]
\REQUIRE A set of $N$ upload bandwidths $u_1\geq ... \geq u_N$\newline
An integer $n_0$ (number of initial copies) \newline
A maximum value $n_{\max}$
\ENSURE $D_1(n)$ for $n\leftarrow$ $1$ to $n_{\max}$

\STATE $L \longleftarrow \text{zeros}(n_0\times 1)$ \label{algd1:initiate}
\FOR{$i\leftarrow 1$ to $n_{\max}$}
\STATE $D_1(i) \longleftarrow \min(L)$ \label{algd1:compi}
\STATE $L = L \setminus \{D_1(i)\}$\label{algd1:remove}
\IF{($i\leq N \And u_i>0$)}
\STATE $L_i = D_1(i) + \{\frac{1}{u_i},\ldots,\frac{n_{\max}-i}{u_i} \}$ \label{algd1:compadd}
\STATE $L=L\cup L_i$
\ENDIF
\ENDFOR
\STATE \textbf{return} $D_1$
\end{algorithmic}
\end{algorithm}

The idea of Algorithm~\ref{algo:r2sw} is that if one computes the times when a new copy of the chunk can be made available, greedy dissemination is always optimal for a single chunk transmission: at any time when a chunk copy ends, if the receiver of that copy is not the best peer missing the chunk, it reduces the usable bandwidth and therefore increases the delay. So the algorithm maintains a \emph{time-completion} list that indicates when copies of the chunk can be made under a bandwidth-greedy allocation. In details:
\begin{itemize}
	\item at line \ref{algd1:initiate}, the completion time list is initiated with $n_0$ values of $0$ (the $n_0$ primary copies);
	\item line \ref{algd1:compi} chooses the lowest completion available completion time and allocates the corresponding chunk copy to the best non-capable peer $i$;
	\item at line \ref{algd1:remove}, the corresponding value $D_1(i)$ is removed, without multiplicity;
	\item the times when $i$ can transmit chunks are added to the list at line \ref{algd1:compadd}.
\end{itemize}

\paragraph*{Remark} in~\cite{liu07minimum}, Liu proposed a snowball approach for computing a feasible delay. The difference between Liu's algorithm and ours is that Liu used a greedy scheduling based on the time when a peer is  to \emph{start} a chunk transmission, while we use the time when it is able to \emph{finish} a transmission. As a result, our algorithm gives the exact value of $D_1$, but the price is that the corresponding scheduling is not practical: it needs all peer to synchronize according to their respective finish deadlines, while Liu's algorithm only requires that \emph{ready} peers greedily select a destination peer. Also note that although Algorithm~\ref{algo:r2sw} provides the exact value for $D_1$, the actual behavior of the delay is difficult to analyze. In the following, we propose to give explicit bounds for $D_1$.

\begin{conjecture}
\label{conj}
The following bounds hold for $D_1$:
\begin{equation}
D_m\leq D_1< \frac{n_0}{U_{n_0}}+\frac{D_m}{\ln(2)}
\label{eq:d1_sharp}
\end{equation}
\end{conjecture}

This conjecture expresses the fact that the price for forfeiting the multi-sources capacities (leaving the many-to-one model for the one-to-one model) is a delay increase that is up to a factor $\frac{1}{\ln(2)}$ and some constant.

\begin{proof}[Proof in the homogeneous case]
The left part of the inequality only expresses that $D_m$ is an absolute lower bound for chunk-based diffusion. For the right part, as stated by Equation~\eqref{eq:delay_dm}, we have $D_m(n)=\sum_{n_0}^{n-1}\frac{1}{ku}\geq \frac{1}{u}\ln(\frac{n}{n_0})$. On the other hand, as stated for instance in~\cite{liu07minimum}, $D_1$ is given by $D_1(n)=\frac{1}{u}\lceil\log_2(\frac{n}{n_0})\rceil$. We deduce
$$\begin{array}{rl}
D_1(n) & <	\frac{1}{u}(\log_2(\frac{n}{n_0})+1)=\frac{n_0}{n_0u}+\frac{1}{u}\frac{\ln(\frac{n}{n_0})}{\ln(2)}\\
& \leq \frac{n_0}{U_{n_0}}+\frac{D_m}{\ln(2)}
\end{array}$$
\end{proof}

To complete the proof, we should show that if we start from a homogeneous system and add some heterogeneity into it, the bounds of Equation~\eqref{eq:d1_sharp} still holds. This is confirmed by our experiments, which show that the homogeneous scenario is the one where the $\frac{n_0}{U_{n_0}}+\frac{D_m}{\ln(2)}$ bound is the tightest. In fact, it seems that the more heterogeneous a system is, the more the behavior of $D_1$ is close to $D_m$. We aim at providing a complete, rigorous proof of Conjecture~\ref{eq:d1_sharp} in a near future work.

\paragraph*{Remark} a less tight, yet easier to prove, relationship between $D_1$ and $D_m$ is
\begin{equation}
D_1< \frac{n_0}{U_{n_0}}+2D_m\text{.}
\label{eq:loosedmd1}
\end{equation}
This inequality comes from the fact that at any given moment, the quantity of \emph{raw data} present in the system (the sum of the complete chunks copies and of the partially transferred chunks) is no more than twice the amount of complete copies: this is straightforward by noticing that for each partially downloaded copy, one can associate a complete, distinct, one (owned by the sender of that copy). The additive constant $\frac{n_0}{U_{n_0}}$ insures that a quantity $2n_0$ of data is present in the system. The $2$ factor comes from the fact that after a time $\frac{2}{U_n}$, a raw quantity of at least $2n$ (more than $n$ complete copies) becomes at least $2(n+1)$ (more than $n+1$ complete copies).

In rest of the paper, however, we prefer to use the conjectured Equation~\eqref{eq:d1_sharp} instead of Equation \eqref{eq:loosedmd1} because of its tightness.

\subsubsection{\texorpdfstring{Properties of $D_1$}{Properties of D1}}

Most of the properties observed for the $(\infty/1)$ model have an equivalent in the $(1/1)$ model. This equivalent can be obtained using Conjecture~\ref{conj}. For instance, the gain of heterogeneity is given by combining Equations \eqref{eq:gain_dm} and \eqref{eq:d1_sharp}:
\begin{equation}
D_1(n)<\frac{1}{\bar{u}}+D^{\bar{u}}_1(n)\text{.}
\label{eq:gain_d1}
\end{equation}

In other words, \emph{up to some constant}, an heterogeneous system is faster than an equivalent homogeneous system. However, this constant means the delay can actually be higher. For instance, consider the four peer system with $(u_1,u_2,u_3,u_4)=(1.6,0.8,0.8,0.8)$ and $n_0=2$. It is easy to verify that $D_1(4)=1.25$ for that particular system, whereas for the equivalent homogeneous system (all peers' bandwidths equal to one) we have $D_1(4)=1$. This is a good illustration of the fact that because of quantification issues, heterogeneity is not always a blessing in the $(1/1)$ model. 

\subsubsection{\texorpdfstring{Extension to $(1/c)$ systems}{Extension to (1/c) systems}}

All the results of the $(1/1)$ systems can be straightforwardly extended to $(1/c)$ ones. Remember that the only difference is that instead of being able to sent one copy to one chunk every $\frac{1}{u_i}$ seconds, a peer $i$ can fetch up to $k$ peers with the chunk every $\frac{c}{u_i}$ seconds. In fact the only reason we have studied $(1/1)$ separately was that $(1/1)$ is a fulcrum model, more commonly used than the generic $(1/c)$ one, so we wanted to highlight it in order to clearly separate the impact of disabling multi-source capabilities and from the possibility of using parallelism.

As the reasonings are mostly the same than for the $(1/1)$, we propose to directly state the results. First, the exact value of $D_c$ can be computed by a slight modification of Algorithm~\ref{algo:r2sw}: all that is needed is to rename $D_1$ to $D_c$ and replace the line~\ref{algd1:compadd} by
\begin{algorithmic}
\STATE $$\begin{array}{rl}
L_i = D_c(i) + \{&\underbrace{\frac{c}{u_i},\ldots,\frac{c}{u_i}}_{c\text{ times}},\underbrace{\frac{2c}{u_i},\ldots,\frac{2c}{u_i}}_{c\text{ times}},\ldots\\
&\ldots,\underbrace{\frac{\lceil\frac{n_{\max}-i}{c}\rceil c}{u_i},\ldots,\frac{\lceil\frac{n_{\max}-i}{c}\rceil c}{u_i}}_{c\text{ times}}\}\text{.}
\end{array}$$
\end{algorithmic}

Then, the relationship between $D_c$ and $D_m$ is given by the following conjecture:
\begin{conjecture}
\label{conjc}
The following bounds hold for $D_c$:
\begin{equation}
D_m\leq D_1< c\frac{n_0}{U_{n_0}}+\frac{c}{\ln(1+c)}D_m\text{.}
\label{eq:dc_sharp}
\end{equation}
\end{conjecture}

This conjecture expresses the fact that the price for using mono-source and $c$-parallelism, compared to the optimal multi-sources-enabled model, is a delay increase that is up to a factor $\frac{c}{\ln(1+c)}$ (and some constant). It is validated by experience, and proved in the homogeneous case, whereas a bound fully proved for the general case is
\begin{equation}
D_m\leq D_1< c\frac{n_0}{U_{n_0}}+(c+1)D_m\text{.}
\label{eq:dc_mou}
\end{equation}

Lastly, the so-called gain of heterogeneity is still only guaranteed up to some constant:
\begin{equation}
D_c(n)<c\frac{n_0}{U_n0}+\log_c(\frac{n}{n_0})<c\frac{1}{\bar{u}}+D^{\bar{u}}_c(n)\text{.}
\end{equation}

\subsubsection{Example}
\label{subsec:validation}

In order to illustrate the results given in that section, we propose to consider a system of $N=10^{4}$ peers that are fetch with $n_0=5$ initial copies of a chunk. We propose the three following distribution:
\begin{itemize}
	\item a homogeneous distribution $H_0$;
	\item a heterogeneous distribution $H_1$ with $3$ bandwidth classes, and a range factor of $10$ between the highest and the lowest class;
	\item a heterogeneous distribution $H_2$ with $3$ bandwidth classes, and a range factor of $100$.
\end{itemize}

The details of the size and upload capacity of each class are expressed in Table~\ref{tab:bandwidths}. The numbers were chosen so that the average bandwidth is $1$ in the three distributions, so we can say they are equivalent distributions, except for the heterogeneity.

\begin{table}%
\begin{center}
\begin{tabular}{|c|c|c|c|}
\hline
 & $H_0$ (Homogeneous)& $H_1$ (Lightly-skewed) & $H_2$ (Skewed)\\
 \hline
$C_1$ & \multirow{3}{2cm}{~~~~$(100\%,1)$} & $(33\%,2.22)$ & $(30\%,2.92)$ \\
 \cline{1-1} \cline{3-4}
$C_2$ & & $(33\%,0.56)$ & $(40\%,0.292)$ \\
 \cline{1-1} \cline{3-4}
$C_3$ & & $(33\%,0.222)$ & $(30\%,0.0292)$ \\
 \hline
\end{tabular}
\end{center}
\caption{Relative size and upload capacity of the classes of $3$ bandwidth distributions}
\label{tab:bandwidths}
\end{table}

\begin{figure}
	\centering
		\subfloat[$H_0$ (homogeneous distribution)]{\includegraphics[width=.32\textwidth]{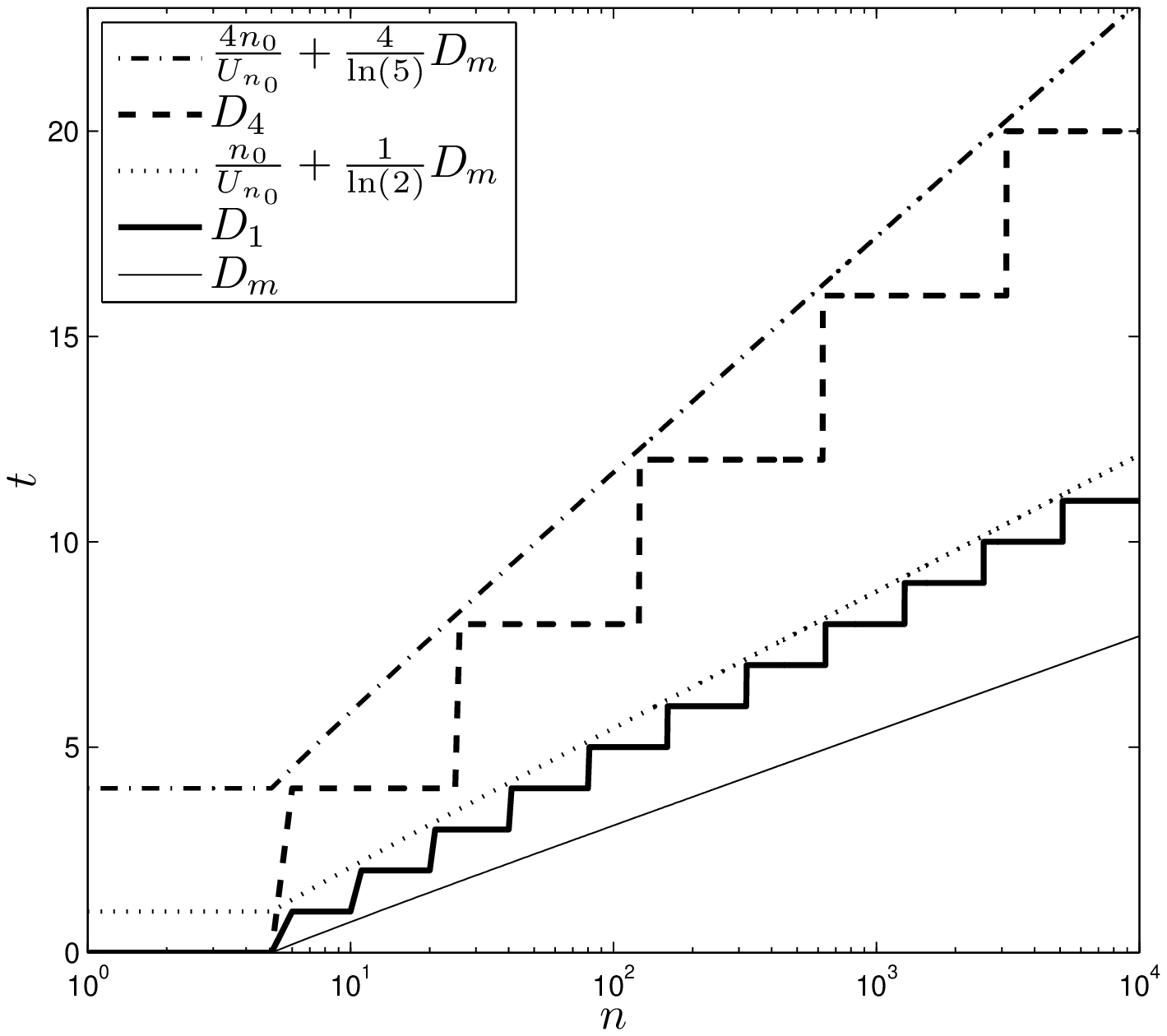}	\label{fig:Ghomo}}
		\subfloat[$H_1$ (Lightly-skewed distribution)]{\includegraphics[width=.32\textwidth]{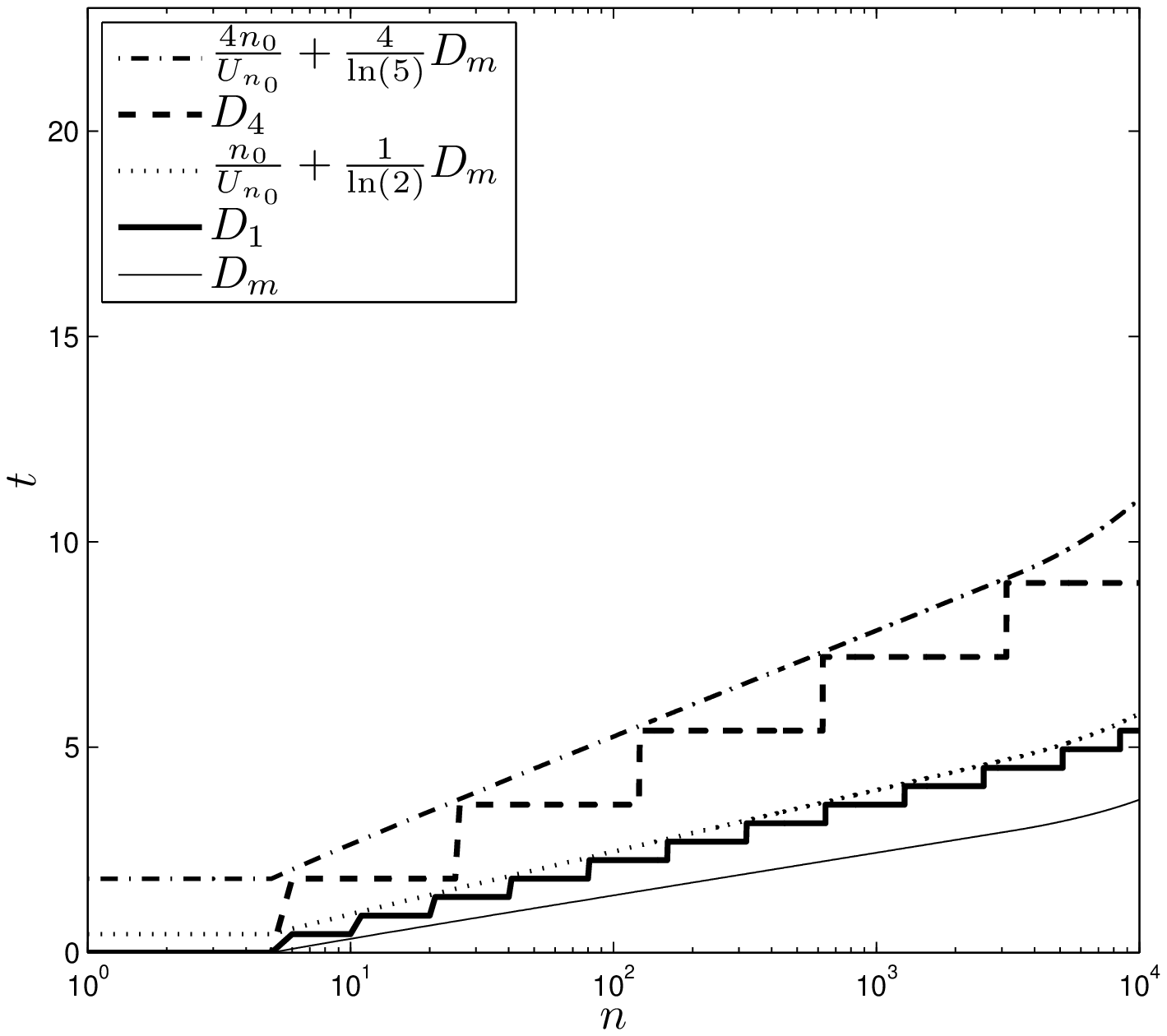}	\label{fig:Gclasses}}
		\subfloat[$H_2$ (Skewed distribution)]{\includegraphics[width=.32\textwidth]{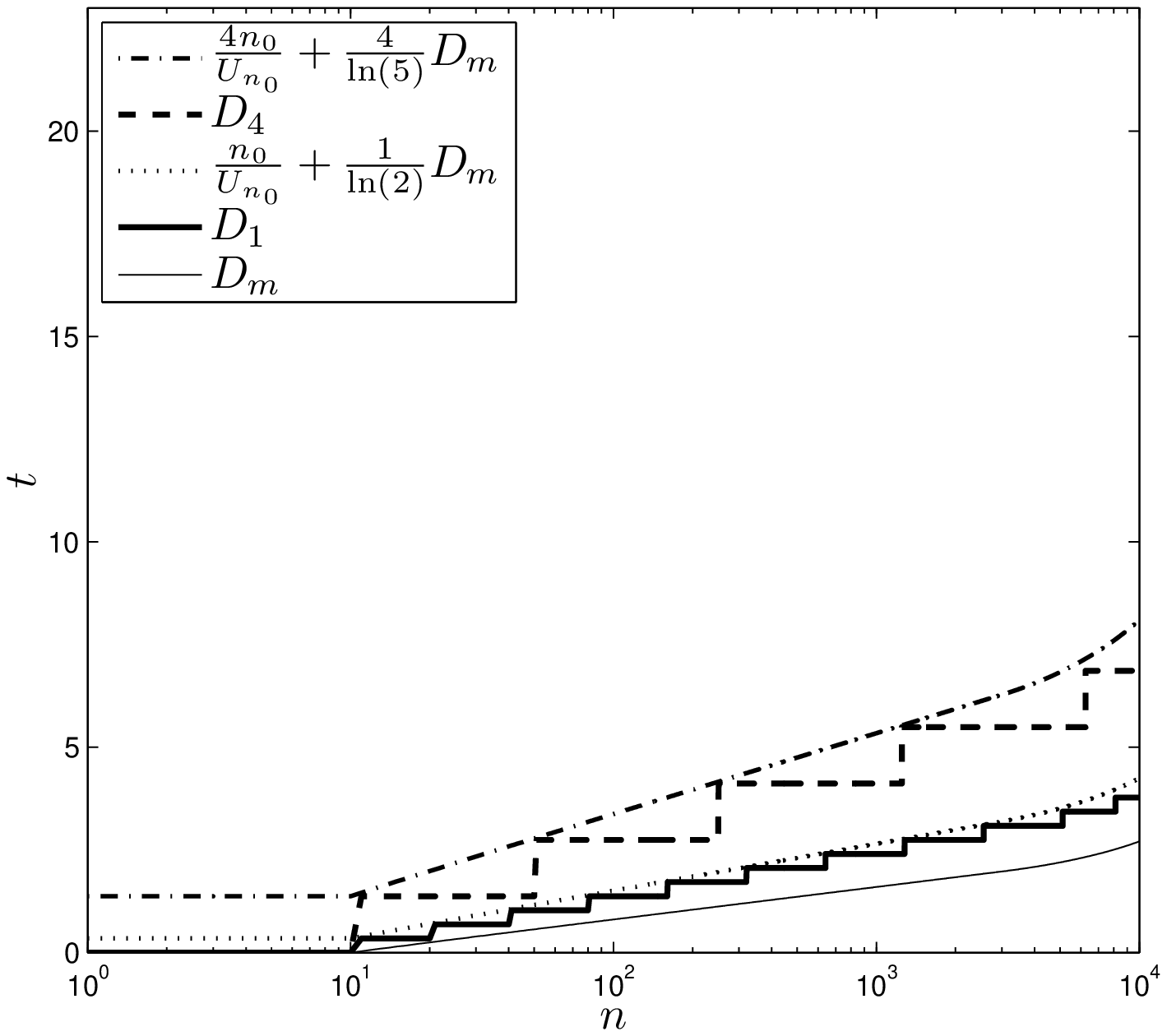}	\label{fig:Gclasses2}}
	\caption{Single chunk diffusion delays for several bandwidth distributions}
	\label{fig:lafig}
\end{figure}

The diffusion delays are displayed in Figure~\ref{fig:lafig}. For each bandwidth distribution, we displayed:
\begin{itemize}
	\item the optimal delay $D_m$, given by Equation \eqref{eq:dmin_multi};
	\item the delays $D_1$ and $D_4$ of the $(1/1)$ and $(1/4)$ models, given by the Algorithm \ref{algo:r2sw} and its modified version;
	\item the upper bounds for $D_1$ and $D_4$ given by Conjectures \ref{conj} and \ref{conjc}.
\end{itemize}

From the observed results, we can say the following:
\begin{itemize}
	\item the delays increase logarithmically for the considered distributions (or equivalently, the chunk diffusion growths exponentially with time), as predicted by Equation~\eqref{eq:dmin_classes}. Note that this logarithmic behavior is only valid for no too skewed distribution: the existence of a highly dominant class may induce an asymptotical linear behavior (cd Equation~\eqref{eq:dominant} and~\eqref{eq:freeriders});
	\item Conjectures \ref{conj} and \ref{conjc} (price of mono-source diffusion and price of parallelism) are verified. Of course, we also confronted these conjectures to a lot of distributions not discussed in this paper (power laws, exponentially distributed, uniformly distributed, with free-riders,\ldots) and they were verified in all cases);
	\item $D_1$ and $D_4$ looks like simple functions. This comes from the fact that we used bandwidth classes, so simultaneous arrivals of new copies is frequent. Nevertheless, $D_1$ and $D_c$ always look less smooth than $D_m$ even for continuous distributions, because the arrival of new chunks, which is very regular in the $(\infty/1)$ model, is more erratic in the $(1/*)$ models;
	\item Delays are faster in $H_2$ than in $H_1$, and faster in $H_2$ than in $H_0$. This is the gain of heterogeneity.
\end{itemize}

\section{Stream of chunks diffusion}
\label{sec:stream}

The issue brought by the stream of chunks problem, compared to the single chunk problem, is that each chunk is in competition with the others for using the bandwidth of the peers: when a peer is devoted to transmitting one given chunk it cannot be used for another one\footnote{An exception is the $(1/c)$ model, however we believe that transmitting different chunk in parallel is not very effective, at least w.r.t. delay.}. Therefore $D$ is a lower bound for $\tilde{D}$, but it is not necessary tight. In this section, we propose to see how $\tilde{D}$ can be estimated.

\subsection{Feasibility of a chunk-based stream}

A first natural question, before studying $\tilde{D}$, is to know whether the stream problem is feasible or not. By adapting a result from~\cite{liu08performance}, we can answer that question.

\begin{theorem}
\label{thm:feasibility}
A necessary, for any diffusion model, and sufficient, for the $(\infty/1)$ and $(1/1)$ models, condition for the stream problem to be feasible is
\begin{equation}
n_0+\frac{1}{s}\sum_{i=1}^Nu_i\geq N
\label{eq:feasibility}
\end{equation}
\end{theorem}

\begin{proof}
The proof is directly derived from Theorem $1$ in \cite{liu08performance} and its proof\footnote{As claimed in \cite{liu08performance}, the technique is in fact inspired by \cite{kumar07stochastic}.}.

Equation \eqref{eq:feasibility} is necessary because it expresses the bandwidth conversation laws: the total bandwidth of the whole system (source and peers) must be greater than the $Ns$ bandwidth needed for the $N$ peers to get the stream.

We then have to show that the condition is sufficient for the $(1/1)$ model. As $(\infty/1)$ can act like $(1/1)$ (the multi-source capacity is not an obligation), this will prove the result for $(\infty/1)$ as well. In the proof in \cite{liu08performance}, the authors constructs a solution where each peer receives from the source a stripe whose rate is proportional to that peer's bandwidth. It is then in charge of distributing that stripe to all other peers. To adapt this to a chunk-based scenario, we follow the same idea: each peer will be responsible for a part of the chunks. We just have to distribute the chunks from the source to the peers according to a scheduler that ensures that the proportion of chunks sent to a given peer is as proportional as possible to its upload bandwidth (for instance, for each new chunk, send it to the peer such that the difference between the bandwidth and the chunk responsibility repartition is minimal). Note that there is situations (case $2$ in the proof in \cite{liu08performance}) where the source must distribute some chunks to all the peers. In those situations, a capacity $1$ of the source is devoted to initial allocation, while the remaining $n_0-1$ capacity is used like a virtual $(N+1)^{th}$ peer (so in those cases, the source may have to handle old chunks in addition of injecting new ones).
\end{proof}

Theorem~\ref{thm:feasibility} basically states that if the bandwidth conservation is satisfied, any chunk-based system is feasible. But while the proposed algorithm is delay-optimal in a stripe-based system, the resulting delay is terrible in a chunk-based system: if $n$ is the label of the last peer with a non-null upload bandwidth, the chunks for which $n$ is responsible (they represent a ratio $\frac{u_n}{U_n}$ of the emitted chunks) needs at least a delay $\frac{N-1}{u_n}$ to be transmitted. In fact, it may need up to $2\frac{N-1}{u_n}$: because of quantification effects, it may receive a new chunk before it has finished the distribution of the previous one. This transmission delay is lower for all other chunks, so the (loose) bound that can be derived from the feasibility theorem is
\begin{equation}
\tilde{D}\leq 2\frac{N-1}{u_n}\text{, for $u_n=\min_{u_i>0}(u_i)$.}
\label{eq:delai_fisibilito}
\end{equation}

\subsection{When heterogeneity is a curse}

One may think that the bound of Equation \eqref{eq:delai_fisibilito} is just a side-effect of the construction proposed in~\cite{liu08performance}, which is not adapted to chunk-based systems. Maybe in practice, as long as the feasibility condition is verified, $\tilde{D}$ is comparable to $D$? This idea is wrong, as shown by the following simple example: for a given $0<\epsilon<\frac{1}{2}$ consider a chunk-based system of two peers with upload bandwidths $u_1=1-\epsilon$ and $u_2=\epsilon$ respectively, $n_0=1$, $s=1$. We have $D_m=D_1=\frac{1}{1-\epsilon}$ ($u_1$ receives the peer and transmits it to $u_2$). Equation~\eqref{eq:feasibility} is verified so the system is feasible. However, when considering the stream problem, $u_1$ alone has not the necessary bandwidth to support the diffusion. Therefore at some point, the source is forced to give a chunk to $u_2$, which need $\frac{1}{\epsilon}$ for sending a chunk. Therefore we necessarily have $\tilde{D}\geq\frac{1}{\epsilon}$, so the minimal achievable delay can be arbitrary great. As a comparison, in the equivalent homogeneous case ($u_1=u_2=\frac{1}{2}$), we have $D=\tilde{D}=2$.

Then again, one could argue that this counter-example of heterogeneity's efficiency is somehow artificial, as only $2$ peers are considered and the available bandwidth is critical. The following theorem proves the contrary.

\begin{theorem}
\label{thm:linear_delay}
Let $n_0\geq 1$, $V\geq 0$, and $s>0$ be fixed. There exist $(1/1)$ systems of size $N$ that verify the following:
\begin{itemize}
	\item the source has capacity $n_0$;
	\item $U_N=\sum_{i=1}^Nu_i \geq Ns +V$ (the system is feasible and the peers have an excess bandwidth of at least $V$);
	\item $\tilde{D_1}=\Omega(N)$.
\end{itemize}
\end{theorem}

Remember that for an homogeneous system, the two first conditions imply $\tilde{D}=O(\log(N))$: for the systems considered by the theorem, heterogeneity is indeed a curse, although the bandwidth is over-provisioned!

\begin{proof}
The idea is exactly the same than for the two-peers example: having peers with a very low upload bandwidth and showing that the system has to use them from time to time. Here we assume $N>n_0+1$ and we consider a system with source capacity $n_0$ and the following bandwidth distribution:
\begin{itemize}
	\item $u_1=(N-n_0-1)s$,
	\item $u_i=\frac{n_0+V+1}{N-1}s$ for $2\leq i \leq N$.
\end{itemize}
By construction, the two first conditions are verified. However, $n_0s+u_1<Ns$, so only the source and $u_1$ do not suffice to distribute the stream. This means that at some point, at least one peer $i>1$ must send at least one chunk to at least one other peer, which takes $\frac{1}{u_i}=\frac{N-1}{s(n_0+V+1)}=\Omega(N)$.
\end{proof}

\subsection{When heterogeneity can be a blessing}

There is at least one case where we know for sure that $\tilde{D}=D$ even for heterogeneous systems: if $D(N)\leq\frac{1}{s}$, then the system can perform the optimal diffusion of a chunk before the next one is injected in the system. There is no competition between different chunks. For instance, in the $(\infty/1)$ model, we have $D_m(N)\leq\frac{1}{\bar{u}}\ln(\frac{N}{n_0})$ (Equation \eqref{eq:gain_dm}), so if $\bar{u}\geq \ln(\frac{N}{n_0})s$, we have $\tilde{D_m}=D_m(n)$.

Of course, this implies a tremendous bandwidth over-provisioning that makes this result of little practical interest. However, the idea can lead to more reasonable conditions, as shown by the following theorem.

\begin{theorem}
\label{thm:transverse}
For a given $(\infty/1)$ system, if one can find an integer $E$ that verifies:
\begin{enumerate}
	\item the $(\infty/one)$ single-chunk transmission delay of the sub-system made of the peers $E,2E,\ldots,\lfloor\frac{N}{E}\rfloor E$ is smaller than $\frac{E}{s}$,
	\item $\bar{u}\geq s+ E\frac{U_{E-1}}{N}$,
\end{enumerate}
then we have $\tilde{D_m}\leq 2\frac{E}{s}$.
\end{theorem}

The second condition is about bandwidth provisioning, whereas the first condition is called the non-overlapping condition (cf the proof below). Of course, $E$ should be chosen as small as possible.

\begin{figure}%
\centering
\includegraphics[width=.6\textwidth]{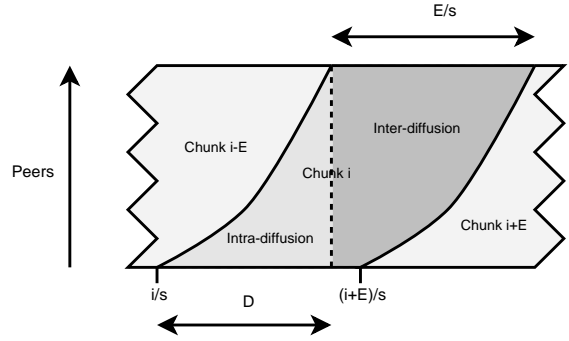}%
\caption{Principle of the \emph{intra-then-inter} chunk distribution}%
\label{fig:stream}%
\end{figure}

\begin{proof}
The idea is to construct a scheduling algorithm that \emph{protects} each chunk, so that it can be optimally diffused, at least for a few moments after it is injected. For that purpose, we split the peers into $E$ groups of peers $G_1,\ldots,G_E$, such that group $G_g$ contains all peers $i$ that verify $i\equiv g \pmod{E}$. Then we use the following \emph{intra-then-inter} diffusion algorithm, whose principle is illustrated in Figure~\ref{fig:stream}. For a given chunk $i$, we do the following
\begin{itemize}
	\item the source injects the chunk $i$ to the $n_0$ best peer of the group $G_g$ that verifies $i\equiv g \pmod{E}$. If $n_0>|G_g|$, the extra copies are given to peers from other groups;
	\item chunk $i$ is diffused as fast as possible inside the group $G_g$. This intra-diffusion ends before the next chunk $i+E$ is sent to $G_g$;
	\item as soon as the intra-diffusion is finished (we call $D_g$ the required time), all peers of $G_g$ diffuse the chunk $i$ to the other groups (inter-diffusion). Of course each peer of $G_g$ must cease to participate to the inter-diffusion of $i$ at the moment where it is involved in the intra-diffusion of $i+E$.
\end{itemize} 

If the algorithm works, the diffusion delay of each chunk is bounded by $2\frac{E}{s}$ (cf Figure \ref{fig:stream}), which proves the theorem. This requires first that the intra-diffusion of chunk $i$ is finished before chunk $i+E$ is injected (non-overlapping condition). The slower group is $E$, so the condition is verified is the single-chunk transmission delay of $G_E$ is smaller than $\frac{E}{s}$.
Then we must guarantee that $G_g$ has enough available bandwidth for diffusing the chunk to the other groups. The peers of $G_g$ can send a quantity $\frac{E}{s}\sum_{k=0}^{|G_g|-1}u_{g+kE}$ of chunk $i$, counting both the intra and inter diffusions. This leads to the bandwidth provisioning condition $\frac{E}{s}\sum_{k=0}^{|G_g|-1}u_{g+kE}\geq N-n_0$. By noticing that $\sum_{k=0}^{|G_g|-1}u_{g+kE}\geq \frac{U_N}{E}-U_{E-1}$, we get the bandwidth provisioning condition of the theorem.

\end{proof}

\subsubsection{\texorpdfstring{Extension to the $(1/c)$ model}{Extension to the (1/c) model}}

the equivalent of Theorem \ref{thm:transverse} for the $(1/c)$ model (including $c=1$) is the following:
\begin{theorem}
\label{thm:transverse1c}
For a given $(1/c)$ system, if one can find an integer $E$ that verifies
\begin{enumerate}
	\item the $(1/c)$ single-chunk transmission delay of the sub-system made of the peers $E,2E,\ldots,\lfloor\frac{N}{E}\rfloor E$ is smaller than $\frac{E}{s}$,
	\item $\bar{u}\geq s(1+\frac{c}{E})+E\frac{U_{E-1}}{N}$,
\end{enumerate}
then we have $\tilde{D_c}\leq 2\frac{E}{s}$.
\end{theorem}

\begin{proof}
The proof is almost the same than for the previous theorem. The only difference are the following:
\begin{itemize}
	\item regarding the diffusion algorithm, each peer must start the inter-diffusion at the moment it is not involved in the intra-diffusion any more (in the $(\infty/1)$ model, all peers finish at the same time, but not here so bandwidth would be wasted if all peers wait for the end of the intra-diffusion);
	\item also, when a peer has not the time to transmit a chunk $i$ to other groups before it should be involved in the intra-diffusion of chunk $i+E$, it stays idle until that moment, for avoiding to interfere with the next intra-diffusion;
	\item as a result, a possible quantity of bandwidth may be wasted during the diffusion of $i$. However, the corresponding quantity of data is bounded by $c|G_g|$, which leads to the supplementary $\frac{c}{E}$ term in the bandwidth provisioning condition.
\end{itemize}
\end{proof}

\subsubsection{Example}

\begin{table*}[ht]%
\begin{center}
\begin{tabular}{|c|c|c|c|}
 \multicolumn{4}{c}{$H_0$ (Homogeneous)}\\
 \hline
 & $D$ & $\tilde{D}$ ($s=.9$) & $\tilde{D}$ ($s=.5$)\\
 \hline
 $(\infty/1)$ & $7.70$ & \multicolumn{2}{|c|}{N/A}\\
 \hline
  $(1/1)$ & \multicolumn{3}{|c|}{$11$}\\
 \hline
 $(1/4)$ & \multicolumn{3}{|c|}{$20$}\\
\hline
 \multicolumn{4}{c}{$H_1$ (Lightly-skewed)}\\
 \hline
 & $D$ & $\tilde{D}$ ($s=.9$) & $\tilde{D}$ ($s=.5$)\\
 \hline
 $(\infty/1)$ & $3.72$ & $8.16$ & $9.72$ \\
 \hline
  $(1/1)$ & $5.40$ & $16.51$ & $11.40$\\
 \hline
 $(1/4)$ & $9.00$ & $53.44$ & $19$\\
\hline
\multicolumn{4}{c}{$H_2$ (Skewed)}\\
 \hline
 & $D$ & $\tilde{D}$ ($s=.9$) & $\tilde{D}$ ($s=.5$)\\
 \hline
 $(\infty/1)$ & $2.70$ & $6.04$ & $6.96$ \\
 \hline
  $(1/1)$ & $4.11$ & $14.88$ & $10.11$ \\
 \hline
 $(1/4)$ & $6.86$ & $51.30$ & $16.86$ \\
 \hline
\end{tabular}
\end{center}
\caption{Delay performance examples for the three bandwidth distributions described in Table \ref{tab:bandwidths}}
\label{tab:numbers}
\end{table*}

in order to illustrate previous theorems with real numbers, we consider the three scenarios used in Section~\ref{subsec:validation}. Table~\ref{tab:numbers} gives the single chunk diffusion delays, as well as the upper bounds for $\tilde{D}$ in slightly overprovisioned ($s=0.9$) and a well overprovisioned ($s=0.5$) scenarios. Note that we choose the parameters so that a proper integer $E$ can be found in all cases.

Our main findings are the following:
\begin{itemize}
	\item for the $(\infty/1)$ model, it is easy to find an integer $E$ close to $sD_m$. This leads to a good delay performance, which for the distribution $H_2$ is better than the delay of the homogeneous case;
	\item for $(1/1)$ and $(1/4)$, the $\frac{c}{E}$ term in the overprovisioning condition can require to pick a high value of $E$ for that condition to be verified, leading to large delays. This is especially noticeable for the $(1/4)$ model and $s=0.9$;
	\item as a result, for these mono-source models, the bounds are not better that the known streaming delays in the homogeneous. Of course, this is not a proof that heterogeneity is a curse in that case: it may exist diffusion schemes that achieves lower streaming delays. But such schemes may be hard to find (and heterogeneity may be considered as a curse in that sense).
\end{itemize}

\section{Conclusion}
\label{sec:conclu}

We investigated the performance of heterogeneous, chunk-based, distributed live streaming systems. We started by studying the transmission of one single chunk and showed that heterogeneous systems tends to produce faster dissemination than equivalent homogeneous systems. We then studied the transmission of a stream of chunks, where heterogeneity can be a disadvantage because the coordination between concurrent chunk diffusions is more complex than for the homogeneous case. Although there is examples where the feasible delay can be arbitrary long, we gave sufficient conditions to link the feasible stream delay to the single-chunk transmission delay. Because of quantification effects, however, the obtained bounds may require the bandwidth to be highly heterogeneous and/or overprovisioned in order to be competitive with homogeneous scenarios, especially for the mono-source models.

\section*{Acknowledgment}

This work has been supported by the Collaborative Research Contract Mardi II between INRIA and Orange Labs, and by the European Commission through the NAPA-WINE Project, ICT Call 1 FP7-ICT-2007-1, Grant Agreement no.: 214412.

\bibliographystyle{abbrv}
\bibliography{RR-OL-2009-09-001}

\end{document}